\documentclass[11pt]{article}
\usepackage{amsmath,amsthm,amsfonts,amssymb,verbatim}
\usepackage{color,times}
\usepackage{hyperref}

\newcommand{\lca}{{\rm lca}}
\newcommand{\diam}{{\rm diam}}

\newcommand{\C}{\mathcal{C}}

\newcommand{\M}{\mathcal{M}}
\newcommand{\N}{\mathcal{N}}

\newcommand{\R}{\mathbb{R}}

\newcommand{\e}{\varepsilon}

\newcommand{\um}{{\rm UM}}

\newcommand{\comp}{{\rm comp}}
\newcommand{\XX}{\psi} 
\newcommand{\xx}{\psi} 

\newcommand{\E}{{\mathbb{E}}}

\newcommand{\mommit}[1]{}

\theoremstyle{plain}
\newtheorem{theorem}{Theorem}
\newtheorem{lemma}{Lemma}
\newtheorem{corollary}[lemma]{Corollary}

\newtheorem{proposition}[lemma]{Proposition}

\theoremstyle{definition}
\newtheorem{definition}{Definition}

\newcommand{\bR}{{\bar{R}}}
\newcommand{\bP}{{\bar{P}}}
\newcommand{\bQ}{{\bar{Q}}}
\newcommand{\bsize}{{w^{*}}}
\newcommand{\dist}{{\rm dist}}

\newcommand{\ddiam}{{\Delta}}
\newcommand{\DDelta}{{\hat{\Delta}}}

\newcounter{tmp}

\begin{document}

\title{Advances in Metric Ramsey Theory and its Applications \\
\footnote{This is paper is still in stages of preparation, this version is not intended for distribution. A preliminary version of this article was written by the author in 2006, and was presented in the 2007 ICMS Workshop on Geometry and Algorithms \cite{B07-icms}. The basic result on constructive metric Ramsey decomposition and metric Ramsey theorem has also appeared in the author's lectures notes, e.g. \cite{B11-course}. } }

\date{}

\author{Yair Bartal\thanks{School of Engineering and Computer Science, Hebrew
University, Israel.
        Email: yair@cs.huji.ac.il.
    Supported in part by a grant from the Israeli Science Foundation
(1817/17).
    }
    }

\maketitle

\begin{abstract}
  Metric Ramsey theory is concerned with finding large well-structured subsets of more complex metric spaces. For finite metric spaces this problem was first studies by Bourgain, Figiel and Milman \cite{bfm}, and studied further in depth by Bartal et. al \cite{BLMN03}. In this paper we provide deterministic constructions for this problem via a novel notion of \emph{metric Ramsey decomposition}. This method yields several more applications, reflecting on some basic results in metric embedding theory.

 The applications include various results in metric Ramsey theory including
the first
\emph{deterministic}
construction yielding Ramsey theorems with tight bounds, 
a well as stronger theorems and properties, implying
appropriate distance oracle applications.

In addition, this decomposition provides the first deterministic Bourgain-type embedding of finite metric spaces into Euclidean space, and an optimal multi-embedding into ultrametrics, 
thus improving its applications in approximation and online algorithms.

The decomposition presented here, the techniques and its consequences have already been used in recent research in the field of metric embedding for various applications.
\end{abstract}

\thispagestyle{empty}
\newpage
\setcounter{page}{1}

\section{Introduction}\label{introduction}

Metric embedding has played an important role in recent decades in the development of algorithms
and efficient data structures, within Computer Science, and has also contributed 
significantly to the development of fundamental mathematical tools (see \cite{indyk,linial,indyk_dist} for surveys).

 Metric Ramsey Theory is asking whether complex metric spaces contain 
large well-structured subspaces, which embed with low distortion into 
a class of special metric spaces, e.g. Euclidean space. A basic theorem 
in this field \cite{BLMN03,MN06} states that this is indeed true when the 
target class is that of ultrametric spaces (which are in particular Euclidean).

In this paper we provide the first deterministic tight version of this theorem.
Our method is based on a novel construct called \emph{metric Ramsey decomposition}  
which we apply to obtain novel versions of this theorem its applications, as well as reflecting on other basic results in metric embedding theory.

	\begin{itemize}
\item {\bf Metric Ramsey Decompositions:} Our main contribution is a novel tool which can viewed as a deterministic counterpart to
padded probabilistic partitions, a fundamental tool in many metric embedding results (e.g.,
\cite{B96,rao,FRT03,B98,KLMN04,ABN06}).
The method of our the
construction is very simple and combines ideas from
\cite{B04,BBM06,BLMN03,MN06} yielding a construction that is
\emph{elementary} and yet very powerful.
Indeed, we show that this basic notion has many applications related to difference aspects of metric embedding theory.
Essentially, they can replace
probabilistic partitions in several fundamental metric embedding
applications yielding first explicit deterministic and efficient
constructions of embeddings into $L_p$ and ultrametrics. In
particular, we can obtain explicit constructions of embeddings of
\cite{MN06,ABN15,ABN08,ABN07-stoc}. Moreover, we obtain
several new and improved embedding results.

\item {\bf Metric Ramsey Theorems and Distance Oracles:} The decomposition naturally arises within the context of the {\sl metric Ramsey
problem}:  Given an arbitrary metric space,
the goal is to find a large subspace that is
highly structured, that is a subspace which embeds with low
distortion into some natural class of highly structured metric
spaces. Of particular interest is the class of ultrametrics (in
particular, they embed isometrically in Euclidean space). 
This problem has been first addressed by Bourgain, Figiel and Milman
\cite{bfm}, motivated by its relation to Dvoretzky's theorem. 
The problem has been further studied
in a sequence of papers partially motivated by
applications in computer science
\cite{krr,bkrs,BBM06}.
Bartal, Linial, Mendel, and Naor \cite{BLMN03} 
obtained nearly tight bounds 
and Mendel and Naor \cite{MN06} gave a
randomized construction which obtained the asymptotically tight
bound for large distortions. The theorem states that every $n$
point metric space contains a subspace of size $n^{1-1/t}$ which
embeds in an ultrametric with distortion $O(t)$, $t>1$.

 In this paper we provide the first \emph{deterministic}
construction with tight bounds improving on the previously best known bounds of \cite{BLMN03}. Our construction provides distortion of $8t$ (for integer values), which is nearly the best known bound via a deterministic construction. Subsequent to our work a similar construction was shown to obtain distortion $8t-2$ \cite{ACEFN18}. The best probabilistic construction (also achieved following this work) obtains distortion $2e t$ \cite{NT12}.


\item {\bf Strong metric Ramsey theorems:}

Our constructions possesses 
additional stronger properties not provided by past constructions.
\begin{itemize}
\item {\sl Linear subspaces with constant $\ell_q$-distortion:}
One of our main contributions are Ramsey theorems with small
average distortion and $\ell_q$-distortion. In \cite{ABN06} it was
proved that every finite metric space embeds in Euclidean space
with constant $\ell_q$-distortion, for all $q<\infty$ (and the
worst case is $O(\log n)$). In \cite{ABN-soda} is was shown that
every finite metric space embeds into an ultrametric with constant
average distortion and $O(\sqrt{\log n})$ $\ell_2$-distortion (and
the worst case is $O(n)$).

Here we prove that every finite metric space contains a
\emph{linear} size subspace which embeds in an \emph{ultrametric}
with \emph{constant} $\ell_q$-distortion, for all $q<\infty$ (and
the worst case is $O(\log n)$).

\item {\sl Linear subspaces with local distortion:}

In \cite{ABN15-local} local embeddings of metric spaces where
introduced. It is shown there how to achieve local Ramsey
theorems. We give here explicit constructions of these theorems
and prove that they work for arbitrary metric spaces.\footnote{The
randomized constructions that were presented in \cite{ABN15-local}
needed a growth rate assumption on the space.}

\item {\sl Doubling and planar metrics:} We give new Ramsey theorems for decomposable metric spaces into
$L_p$. This is the first example of a metric Ramsey theorem for a
non-trivial family of metric spaces which obtains embeddings into
$L_p$ that beat the best possible bounds given by embedding into
ultrametrics.
\end{itemize}

\item {\bf Proximity Data Structures:}
The metric Ramsey problem is closely related to the construction
of proximity data structures including approximate \emph{distance
oracles} \cite{TZ05}. These are space efficient data structures
that enable satisfying fast approximate distance queries. Our
results yield the {\em first} deterministic construction of
approximate distance oracles with \emph{asymptotically optimal}
space-distortion tradeoff and constant query time.

Our stronger metric Ramsey theorems also translate to new
approximate distance oracles. In particular we provide distance
oracles of \emph{linear} size with \emph{constant}
$\ell_q$-distortion, for all $q<\infty$ (and the worst case is
$O(\log n)$). Similarly, we obtain distance oracles of
\emph{linear} size with local stretch properties. In particular,
we obtain the first construction of such data structure with
$O(n)$ storage and $O(\vartheta(\log k))$ $k$-local stretch (i.e.
this bounds the stretch for the $k$th nearest neighbor of a
point).

 We note that while constructing distance oracles using the basic
 metric Ramsey theorem can be done using the approach of
 \cite{MN06} their method \emph{does not} apply to the case that
 the distortion varies for different pairs of points as is the
 case with our strong metric Ramsey theorems. We therefore
 show directly how to apply our Ramsey decomposition to obtain \emph{Ramsey
embeddings}. This is a strengthening of the Ramsey type results
which we define and have immediate application to proximity data
structures.

\item {\bf Deterministic Embedding:} We provide a unified
framework for \emph{deterministic} Bourgain-type embedding of arbitrary
metrics into $L_p$. This is achieved by applying our deterministic decomposition 
in place of the standard probabilistic padded decompositions. Moreover, builing
ont the embedding of \cite{ABN06} this method implies a deterministic
embedding in $O(\log n)$ dimension and distortion and
\emph{constant} $q$ moments, for all fixed $q<\infty$. This further provides optimal prioritized embeddings \cite{EFN15,BFN16}.

\item {\bf Multi-Embedding:} In a multi-embedding \cite{BM04} a metric space is embedded
into a \emph{larger} metric space so that the distortion of paths is
preserved. These type of embeddings have applications in the
context of approximation and online problems (in particular, the
extensively studied metrical task systems and the group Steiner
tree problems). Using the Ramsey decomposition we obtain \emph{optimal}
multi-embedding into ultrametrics, thereby improving
bounds in the applications.
\end{itemize}

\subsection{Further Related Work and Applications}
The decomposition lemma, Ramsey theorems and methods provided in this paper have been basis for some further results in several papers. Bartal, Fandina and Neiman \cite{BFN19-covers} study tree covers, bounding the number of trees necessary so that every pairwise distance is maintain in one of the trees within a given distortion bound. Among other results they show that Theorem~\ref{thm:ramsey-ultrametrics} implies near tight bounds on Ramsey tree covers. 
Recently, Filtser and Le \cite{FL21} make use of our decomposition and techniques to establish a distributional version of the multi-embedding theorem and thereby obtain a bound on the expected duplicity of points, which they then use to obtain certain ``clan embedding" properties, which they show useful for compact routing. 
In the work of Abraham et. al \cite{ACEFN18} similar though somewhat more involved techniques are applied to obtain Ramsey metric theorems for spanning trees in graphs.

\subsection{Preliminaries}

\begin{definition}
 Let $X,Y$ be metric spaces. An embedding of $X$ into $Y$ is a function $f:X\to Y$.
 The distortion of $f$ for the pair $\{u,v\} \in {X\choose 2}$
is $\dist_f(u,v) = \frac{d_Y(f(u),f(v))}{d_X(u,v)}$. The
distortion of $f$ is given by $\frac{\max_{\{u,v\} \in {X\choose
2}} \dist_f(u,v)}{\min_{\{u,v\} \in {X\choose 2}} \dist_f(u,v)}$.
\end{definition}

An embedding $f$ is \emph{non-contractive} if for any $u, v \in
X$: $d_Y(f(u),f(v)) \geq d_X(u,v)$ and \emph{non-expansive} if for
any $u, v \in X$: $d_Y(f(u),f(v)) \leq d_X(u,v)$.

For a vertex $v$ and $r \ge 0$, the ball at radius $r$ around $v$
is defined as $B(v,r) = \{u \in V | d(u,v) \leq r \}$.

\subsection{Metric Ramsey Theorems}

\subsubsection{Definitions}
We recall some definitions and notation from \cite{BLMN03} that is
useful in the context of the metric Ramsey problem.

\begin{definition}[Metric Ramsey functions]
Let $\M$ be some class of metric spaces. For a metric space $X$,
and $\alpha \geq 1$, $R_\M(X;\alpha)$ denotes the largest size of
a subspace $Y$ of $X$ such that $Y$ embeds in a metric space in
$\M$ with distortion $\alpha$.

Denote by $R_{\M}(\alpha,n)$ the largest integer $m$ such that any
$n$-point metric space has a subspace of size $m$ that
$\alpha$-embeds into a member of $\M$. In other words, it is the
infimum over $X$, $|X|=n$, of $R_\M(X;\alpha)$.


In the most general form, let $\N$ be a class of metric spaces and
denote by $R_\M(\N;\alpha,n)$ the largest integer $m$ such that
any $n$-point metric space in $\N$ has a subspace of size $m$ that
$\alpha$-embeds into a member of $\M$. In other words, it is the
infimum over $X \in \N$, $|X|=n$, of $R_\M(X;\alpha)$.
\end{definition}

It is useful to generalize the metric Ramsey problem to weighted
metric spaces. Although the weighted Ramsey problem defined above
is not necessary to obtain our basic results they are useful for
extending them, e.g. by using results from \cite{BLMN03}.

Let a weighted metric space be a pair $(X,w)$, where $X$ is a
metric space and $w:X\to \R^{+}$ is a weight function. For a
subspace $Z \subseteq X$ let $w(Z) = \sum_{x \in Z} w(x)$. For $0
\leq \xx$, let $w^\xx$ denote the weight function defined by
$w^\xx(x) = w(x)^\xx$, for every $x \in X$.

Note that for weight function $w(x) \equiv 1$: $w(Z) = w^\xx(Z) =
|Z|$. We note that although we will use the weighted metric
notation throughout, in most cases the uniform weight can be
assumed.

\begin{definition}[Weighted Ramsey Function]\label{def:XX}
Let $\M,\N$ be classes of metric spaces. Denote by $\XX_\M(\N,
\alpha)$ the largest $0 \le \xx \le 1$ such that for every metric
space $X \in \N$ and any weight function $w:X\to \R^{+}$, there is
a subspace $Y$ of $X$ that $\alpha$-embeds in $\M$ and satisfies:
$w^\xx(Y) \ge w(X)^\xx$.
When $\N$ is the class of all metric spaces, it is omitted from
the notation.
\end{definition}

The following is an immediate consequence of
Definition~\ref{def:XX}.
\begin{proposition}\label{prop:chi-R}
\[ R_\M(\N;\alpha,n) \ge n^{\XX_\M(\N,\alpha)} . \]
In particular,
\[ R_\M(\alpha,n) \ge n^{\XX_\M(\alpha)} . \]
\end{proposition}

Let $\um$ denote the class of ultrametrics.
  It is well-known (c.f. \cite{lemin})
that ultrametrics embed isometrically in $\ell_2$. Therefore
$\XX_{L_2}(\alpha) \ge \XX_{\um}(\alpha)$.

\newcounter{ThmMetricRamseyTheorem}
\setcounter{ThmMetricRamseyTheorem}{\thetheorem}
\def\ThmMetricRamseyTheorem{
For any integer $t \geq 2$,
\[
\XX_{\um}(8t) \ge 1-1/t.
\]
In particular, any $n$-point metric
space contains a subspace of size $n^{1-1/t}$ which embeds in an
ultrametric with distortion $8/t$. 
}

\begin{theorem}\label{thm:ramsey-ultrametrics}
\ThmMetricRamseyTheorem
\end{theorem}

We note that the constant in the theorem above can in fact be
further improved at a the price of proof elegance.

\subsection{Stronger Ramsey Theorems: Subspaces of Linear Size}

Recall the following definitions from \cite{ABN06}:

\begin{definition}[Partial/Scaling Embedding]
For $\epsilon>0$, a $(1-\epsilon)$-partial embedding $f$ has
distortion $\alpha(\epsilon)$ if it is non-contractive and there
exists a set $G_\epsilon \subset {X\choose 2}$ of size at least
$(1-\epsilon) {n\choose 2}$ such that for every $u,v \in
G_\epsilon$, $\dist_f(u,v) \leq \alpha(\epsilon)$.\footnote{A
special type of partial embedding is called {\em coarse} where
$G_\epsilon$ is composed of pairs $u,v$ where $v$ is not in the
ball containing $\epsilon n/2$ points around $u$.} An embedding
has scaling distortion $\alpha(\epsilon)$ if it is $(1-\epsilon)$-
partial for every $\epsilon>0$.
\end{definition}

\begin{definition}[$\ell_q$-distortion]
 For $1 \leq q \leq \infty$,
define the \emph{$\ell_q$-distortion} of an embedding $f$ as:
\begin{equation*}
\dist_q(f) = \| \dist_f(u,v) \|_q^{({\cal U})} = \E[
\dist_f(u,v)^q ]^{1/q} ,
\end{equation*}
where the expectation is taken according to the uniform
distribution ${\cal U}$ over ${X\choose 2}$. The classic notion of
distortion is expressed by the $\ell_\infty$-distortion and the
average distortion is expressed by the $\ell_1$-distortion.
\end{definition}

In \cite{ABN06} the notion of scaling embedding is shown to be
closely related to the $\ell_q$-distortion of the embedding. They
give partial and scaling distortion results for embedding into
$L_p$. In \cite{ABN-soda} it was shown that every finite metric
space embeds into an ultrametric with $O(\sqrt{1/\epsilon})$
scaling distortion (and this bound is tight). Here we prove the
following strengthened Ramsey theorems:

\newcounter{ThmPartialMetricRamseyTheorem}
\setcounter{ThmPartialMetricRamseyTheorem}{\thetheorem}
\def\ThmPartialMetricRamseyTheorem{
For every $\delta >0$ and $\epsilon >0$, any $n$-point metric
space contains a subspace $Y$ of size $\geq \delta n$ such that
$Y$ has a $(1-\epsilon)$-partial embedding into an ultrametric
with distortion $O(\lceil \log_{1/\delta} 1/\epsilon \rceil )$. }

\begin{theorem}[Partial Metric Ramsey Theorem]\label{thm:ramsey-partial}
\ThmPartialMetricRamseyTheorem
\end{theorem}

Let $\vartheta: \mathbb{R}^{+} \to \mathbb{R}^{+}$ be a function
such that $\int_1^\infty \frac{dx}{\vartheta(x)} = 1$. In
particular for any $\xi>0$, we can have $\vartheta(x) \leq
cx\log^{1+\xi}(x)$.

\newcounter{ThmScalingMetricRamseyTheorem}
\setcounter{ThmScalingMetricRamseyTheorem}{\thetheorem}
\def\ThmScalingMetricRamseyTheorem{
For every $\delta >0$, any $n$-point metric space contains a
subspace $Y$ of size $\geq \delta n$ such that $Y$ has a embedding
into an ultrametric with scaling distortion $O(\lceil
\vartheta(\log_{1/\delta} 1/\epsilon)) \rceil$. As a consequence
its $\ell_q$-distortion is bounded by $O(\lceil
\vartheta(q/\log(1/\delta)) \rceil )$. }

\begin{theorem}[Scaling Metric Ramsey Theorem]\label{thm:ramsey-scaling}
\ThmScalingMetricRamseyTheorem
\end{theorem}

 We also provide additional Ramsey theorems that give local
distortion bounds. However, in this context the Ramsey embedding
version of these theorems are more natural and will be discussed
in the subsequent subsection below.

In addition we present below an improved Ramsey theorem for
doubling and excluded-minor metrics.

\subsubsection{Ramsey Theorems for Decomposable Metric Spaces}

Recall that metric spaces $(X,d)$ can be characterized by their
decomposability parameter $\tau_X$ where it is known that $\tau_X
= O(\log \lambda_X)$, where $\lambda_X$ is the doubling constant
of $X$, and for metrics of $K_{s,s}$-excluded minor graphs.
$\tau_X = O(s^2)$.

\newcounter{ThmMetricRamseyTheoremDecomp}
\setcounter{ThmMetricRamseyTheoremDecomp}{\thetheorem}
\def\ThmMetricRamseyTheoremDecomp{
Let $X$ be a metric space. There exists $C>0$ such that for every
$1 \leq p \leq \infty$, and any $\alpha > 1$:
\[
\XX_{L_p}(X,\alpha) \ge 1-C\left(\frac{\tau_X}{\alpha}\right)^p
p\log\left(\frac{\tau_X}{\alpha}\right).
\]
In particular, for every $\e >0$, $X$ contains a subspace of size
$n^{1-\e}$ which embeds in $L_p$ space with distortion
$\tilde{O}(\tau_X/\e^{1/p})$. }

\begin{theorem}[Ramsey-type Theorem for Decomposable Metrics]\label{thm:ramsey-decomp}
\ThmMetricRamseyTheoremDecomp
\end{theorem}

\subsection{Ramsey Embedding, Ramsey Covers and Proximity Data Structures}

We show that our algorithms for computing metric Ramsey
constructions can be applied to obtain new results for proximity
data structures, including distance oracles and approximate
ranking.

To this aim we define the notion of a {\em Ramsey embedding}.

\begin{definition}
\label{def:ramsey-embedding} Given metric spaces $X,Y$. A Ramsey
embedding of $X$ into $Y$ is a pair composed of an embedding
$f:X\to Y$ and a subspace $X' \subseteq X$. A non-contractive
Ramsey embedding has distortion $\alpha$ if for every $x\in X'$
and $y \in X$, $\dist_f(x,y) \leq \alpha$. We call the subspace
$X'$ the core subspace of the embedding.
\end{definition}

We show that our Ramsey theorems can be extended to provide Ramsey
embeddings. These can be further extended to obtain Ramsey covers
which we then use to obtain first deterministic constructions
distance oracles and approximate ranking data structures with
optimal query, stretch and space tradeoffs.

Moreover we give the first construction of such data structures
with $O(n)$ storage, $O(1)$ query time, and $O(1)$ average
distortion and $\ell_q$ distortion for every fixed $q<\infty$ (and
the worst case distortion is $O(\log n)$).

\subsubsection{$k$-Local Embeddings and Data Structures}

\begin{definition}
For $x\in X$ let $r_k(x)$ the minimum $r$ such that $|B(x,r)|\geq
k$. A Ramsey embedding $f:X\to Y$ with core subspace $X'$ has
$k$-local distortion $\alpha$ if it is \emph{non-expansive} and
for every $x,y \in X$, $d_Y(f(x),f(y)) \geq \min\{ d_X(x,y),
r_k(x) \}/\alpha$. We say that $f$ has local scaling distortion
$\alpha(k)$ if it is $k$-local for every $1 \leq k \leq n$.
\end{definition}

We prove the following strengthened Ramsey embedding theorems:

\newcounter{ThmLocalMetricRamseyTheorem}
\setcounter{ThmLocalMetricRamseyTheorem}{\thetheorem}
\def\ThmLocalMetricRamseyTheorem{
For every $\varepsilon>0$ and $k \in \mathbb{N}$, any $n$-point
metric space has a Ramsey embedding with $k$-local distortion
$O(1/\varepsilon)$ and core subspace $Y$ of size $\geq n \cdot
k^{-\varepsilon}$. }

\begin{theorem}[Local Metric Ramsey-type Theorem]\label{thm:ramsey-local}
\ThmLocalMetricRamseyTheorem
\end{theorem}

\newcounter{ThmScalingLocalMetricRamseyTheorem}
\setcounter{ThmScalingLocalMetricRamseyTheorem}{\thetheorem}
\def\ThmScalingLocalMetricRamseyTheorem{
For every $\varepsilon>0$, any $n$-point metric space has a Ramsey
embedding with scaling local distortion $O(\min\{\vartheta(\log
k),1/\varepsilon\})$ and core subspace $Y$ of size $\geq
n^{1-\varepsilon}$. }

\begin{theorem}[Scaling Local Metric Ramsey-type Theorem]\label{thm:ramsey-scaling-local}
\ThmScalingLocalMetricRamseyTheorem
\end{theorem}

By following the same procedure discussed before we can use these
theorems to obtain new deterministic constructions of distance
oracles and approximate ranking data structures with $k$-local
stretch and scaling local stretch respectively. In particular, we
obtain the first construction of such data structure with $O(n)$
storage and $O(\vartheta(\log k))$ $k$-local stretch (i.e. this
bounds the stretch for the $k$th nearest neighbor of a point).

\subsection{Multi-Embedding}

Another application of our Ramsey decomposition is to obtain
optimal multi-embeddings of metric spaces into ultrametrics
\cite{bm}. A multi-embedding of a space $X$ into $Y$ is a mapping
of points in $X$ to sets of points in $Y$. It is desirable that
the size of $Y$ would be small. The path distortion of a
multi-embedding is $\alpha$ if for every path $p$ in $X$ there is
a corresponding path in $Y$ over the images of points in $p$ whose
length is at most $\alpha$ times the length of $p$.

The following theorem gives a tight bound on multi-embedding into
ultrametrics, improving the previous result of \cite{bm}. This
implied improvements to algorithms for the online metrical task
systems problem and for the group Steiner tree problem for metric
spaces with small aspect ratio, and provides simpler algorithms
for these problems.

\newcounter{ThmMultiEmbedding}
\setcounter{ThmMultiEmbedding}{\thetheorem}
\def\ThmMultiEmbedding{
For any metric space on $n$ points and aspect ratio $\Phi$, and any $\e>0$, there
exists a multi-embedding into an ultrametric of size $n^{1+\e}$,
whose path distortion is at most
$
 O( \min\{\log n , \log \Phi \}/\e).
$
}

\begin{theorem} \label{thm:path-approx}
\ThmMultiEmbedding
\end{theorem}

\section{Notation}

For sets
$U,V \subseteq X$ let $d(U,V) = \min \{ d(u,v) | u \in U, v \in V
\}$. Let $\ddiam(X) = \diam(X)$ denote the diameter of $X$.

\subsection{Ultrametrics and Hierarchically Well-Separated Trees}\label{section:lower-preliminaries}

 Recall that an {\em ultrametric} is a metric space $(X,d)$ such that for every $x,y,z\in X$,
$$
d(x,z)\le \max\{d(x,y),d(y,z)\}.
$$

We recall the following definition from \cite{B96}:

\begin{definition}
\label{def:hst} For $k\geq 1$, a $k$-\emph{hierarchically
well-separated tree} ($k$-HST) is a metric space whose elements
are the leaves of a rooted tree $T$. To each vertex $u\in T$ there
is associated a label $\Lambda(u) \ge 0$ such that $\Lambda(u)=0$
iff $u$ is a leaf of $T$. It is required that if a vertex $u$ is a
child of a vertex $v$ then $\Lambda(u)\leq \Lambda(v)/k$ . The
distance between two leaves $x,y\in T$ is defined as
$\Lambda(\lca(x,y))$, where $\lca(x,y)$ is the least common
ancestor of $x$ and $y$ in $T$.

\end{definition}

First, note that an ultrametric and a $1$-HST are identical
concepts. Any $k$-HST is also a $1$-HST, i.e., an ultrametric. Any
ultrametric is $k$-equivalent to a $k$-HST \cite{B96}.

When we discuss $k$-HSTs, we freely use the tree $T$ as in
Definition~\ref{def:hst}, \emph{the tree defining the HST}.



Let $\um$ to denote the class of ultrametrics, and $k$-HST denotes
the class of $k$-HSTs.

\section{Ramsey Decomposition}
\label{sec:ramsey-decomposition}

Define the spherical-weight of $Z \subseteq X$,
$\bsize(Z) = \max_{z\in Z} w(B(z,\ddiam(Z)/4)).$

\begin{lemma}
\label{lem:decomp} Given a metric space $X$, and $0 < \DDelta \leq
\ddiam(X)/2$, and integer $t \geq 2$, then there exists a
partition $(Q,\bar{Q})$ of $X$, and $P \subseteq Q$, such that:
$\ddiam(Q) \leq \DDelta$, $d(P,\bQ) \geq \DDelta / (4t)$, and
\begin{eqnarray*}
w(P) \geq w(Q) \cdot \left( \frac{\bsize(X)}{\bsize(Q)}
\right)^{-1/t}.
\end{eqnarray*}
\end{lemma}
\begin{proof}
Let $v$ be a node that minimizes the ratio
$w(B(v,\DDelta/2))/w(B(v,\DDelta/4))$. We will choose $Q=B(v,r)$
for some $r \in [\DDelta/4,\DDelta/2]$. For $0 \leq i \leq t$, define $Q_i =
B((1+\frac{i}{t})\DDelta/4)$. Clearly there exist some $i>0$ such
that $w(Q_i) \leq w(Q_{i-1})
\left(\frac{w(Q_t)}{w(Q_0)}\right)^{\frac{1}{t}}$. Then we set
$Q=Q_i$ and $P=Q_{i-1}$. Therefore we have that $d(P,\bQ) \geq
\DDelta/(4t)$ and
\begin{eqnarray*}
w(P) & \geq & w(Q) \cdot \left(
\frac{w(B(v,\DDelta/2))}{w(B(v,\DDelta/4)} \right)^{-1/t}.
\end{eqnarray*}

Now, let $u$ be the node that maximizes $w(B(u,\ddiam(Q)/4))$.
Since $\ddiam(Q) \leq \DDelta$ we have that $\bsize(Q) \leq
w(B(u,\DDelta/4))$. Recall that $\bsize(X) \geq
w(B(u,\DDelta/2))$, as $\DDelta \leq \ddiam(X)/2$. By the choice
of $v$ we conclude that
\begin{eqnarray*}
w(P) & \geq & w(Q) \cdot \left( \frac{w(B(u,\DDelta/2))}{w(B(u,\DDelta/4))} \right)^{-1/t}  
 \geq  w(Q) \cdot \left( \frac{\bsize(X)}{\bsize(Q)}
\right)^{-1/t}.
\end{eqnarray*}\qed
\end{proof}

\section{Metric Ramsey Theorems}
\label{sec:ramsey}

\setcounter{theorem}{\theThmMetricRamseyTheorem}
\begin{theorem}
For any integer $t \geq 2$,
\[
\XX_{\um}(8t) \ge 1-1/t.
\]
\end{theorem}
\setcounter{theorem}{\thetmp}
\begin{proof}
Let $X$ be an arbitrary metric space. Let $Z \subseteq X$. We will
construct a subspace $S(Z) \subseteq Z$ and an ultrametric
$U(S(Z))$ recursively as follows: use the decomposition described
in Section~\ref{sec:ramsey-decomposition} with $\DDelta=\ddiam(Z)/2$ to obtain a partition of the graph
$(Q,\bQ)$ and $P \subset Q$ satisfying Lemma~\ref{lem:decomp}. Run
the algorithm on $P$ and $\bQ$ recursively, obtaining subspaces
$S(P)$ and $S(\bQ)$ and ultrametrics $U(S(P))$ and $U(S(\bQ))$
respectively. Let $S(Z) = S(P)\cup S(\bQ)$. The ultrametric $U(Z)$
is constructed by creating a root $r$ labeled with
$\Lambda(r)=\ddiam(S(Z))$ with two children at which we root the
trees defining $U(S(P))$ and $U(S(\bQ))$.

We first prove by induction on the size of $Z \subset X$ that
$S(Z)$ is $4t$ equivalent to $U(S(Z))$ via a
non-contractive embedding and $\ddiam(U(S(Z)) = \ddiam(S(Z))$.

If $Z$ includes a single point $z$ then the claim trivially holds.
Assume by induction that the claim holds for strict subsets of
$Z$. Consider $x,y \in S(Z)$. If $x,y \in S(P)$ their distance in
$U(Z)$ is the same as in $U(P)$ and therefore the claim follows
from the induction hypothesis. If $x,y \in \bQ$ then a similar
argument holds. Let $\eta = 1/(4t)$. Let $x \in P$ and $y \in \bQ$ then by
Lemma~\ref{lem:decomp}, $d(x,y) \geq \eta/2 \cdot 
\ddiam(Z) \geq \eta/2 \cdot\ddiam(S(Z))$. As $d(x,y) \leq
\ddiam(S(Z))$ it follows that $d(x,y) \leq d_{U(Z)}(x,y) \leq
2/\eta\cdot d(x,y)$, and we conclude that $S(Z)$ is
$(8t)$-equivalent to $U(Z)$.

Let $\xx = 1-1/t$. Next, we prove by induction on
the size of $Z \subseteq X$ that
\[
w^\xx(S(Z)) \geq w(Z)\cdot \bsize(Z)^{-1/t}.
\]
If $Z=\{z\}$ includes a single point $z$ then the claim trivially
holds since $w^\xx(S(Z))= w(Z) \cdot \bsize(Z)^{-1/t} =
w(z)^\xx$. By applying the induction hypothesis and
Lemma~\ref{lem:decomp} we obtain
\begin{eqnarray*}
w^\xx(S(Z)) & = & w^\xx(S(P)) + w^\xx(S(\bQ)) \geq w(P)\cdot
\bsize(P)^{-1/t} + w(\bQ)\cdot \bsize(\bQ)^{-1/t} \\
& \geq & w(Q) \cdot \left(\frac{\bsize(Z)}{\bsize(Q)}\cdot
\bsize(P)\right)^{-1/t} + w(\bQ)\cdot
\bsize(Z)^{-1/t} \\
& \geq & (w(Q) + w(\bQ)) \cdot \bsize(Z)^{-1/t} = w(Z) \cdot
\bsize(Z)^{-1/t}.
\end{eqnarray*}
Noting that $\bsize(X) \leq w(X)$ we conclude that $w^\xx(S(X))
\geq w(X)^\xx$.
\end{proof}

\subsection{Stronger Ramsey Theorems: Subspaces of Linear Size}

In this section we will fix $w(Z) = w^\psi(Z) = |Z|$ (however, the
claims and proofs can be appropriately generalized to more general
weight functions).

\setcounter{theorem}{\theThmPartialMetricRamseyTheorem}
\begin{theorem}
\ThmPartialMetricRamseyTheorem
\end{theorem}
\setcounter{theorem}{\thetmp}
\begin{proof}
Let $X$ be an arbitrary metric space. Let $Z \subseteq X$. We
construct a subspace $S(Z) \subseteq Z$ and an ultrametric
$U(S(Z))$ recursively as follows. If $w(Z) \leq \epsilon \cdot
w(X)$ we let $S(Z) = Z$ and define $U(S(Z))$ to be a rooted star
with $|Z|$ leaves and label the root with
$\Lambda(r)=\ddiam(S(Z))$. Otherwise, if $w(Z)
> \epsilon \cdot w(X)$ use the decomposition described in
Section~\ref{sec:ramsey-decomposition} with $\eta = 1/\lceil
6\log_{1/\delta} 1/\epsilon \rceil$ and $\DDelta=\ddiam(Z)/2$ to
obtain a partition of the graph $(Q,\bQ)$ and $P \subset Q$
satisfying Lemma~\ref{lem:decomp}. Run the algorithm on $P$ and
$\bQ$ recursively, obtaining subspaces $S(P)$ and $S(\bQ)$ and
ultrametrics $U(S(P))$ and $U(S(\bQ))$ respectively. Let $S(Z) =
S(P)\cup S(\bQ)$. The ultrametric $U(Z)$ is constructed by
creating a root $r$ labeled with $\Lambda(r)=\ddiam(S(Z))$ with
two children at which we root the trees defining $U(S(P))$ and
$U(S(\bQ))$.

 We first prove by induction on the size of $Z \subset X$ that
there is a $(1-\epsilon)$-partial embedding of $S(Z)$ into
$U(S(Z))$ with distortion $2/\eta = O(\lceil \log_{1/\delta}
1/\epsilon \rceil)$, and that $\ddiam(U(S(Z)) = \ddiam(S(Z))$.

If $w(Z) \leq \epsilon \cdot w(X)$ then the number of pairs in $Z$
is at most ${|Z|\choose 2} \leq |Z| \epsilon (n-1) /2$. All such
pairs are excluded from $G_\epsilon$. It follows that the number
of excluded pairs sums up to at most $\epsilon {n\choose 2}$ in
total. Otherwise if $w(Z) > \epsilon \cdot w(X)$ then the same
argument in the proof of Theorem~\ref{thm:ramsey-ultrametrics}
holds.

Let $\xx = 1-6\eta$. Next, we prove by induction on the size of $Z
\subseteq X$ that
\[
w^\xx(S(Z)) \geq w(Z)\cdot \left\lceil \frac{\bsize(Z)}{\epsilon \cdot
w(X)} \right\rceil^{-6\eta}.
\]
If $w(Z) \leq \epsilon \cdot w(X)$ then $S(Z)=Z$ and the claim
trivially holds since $\bsize(Z) \leq w(Z) \leq \epsilon \cdot
w(X)$. Otherwise if $w(Z) > \epsilon \cdot w(X)$ then essentially
the same argument in the proof of
Theorem~\ref{thm:ramsey-ultrametrics} holds.

 Noting that $\bsize(X) \leq w(X)$ we conclude that
$w^\xx(S(X)) \geq w(X) \cdot \left( \frac{1}{\epsilon}
\right)^{-6/\lceil 6\log_{1/\delta} 1/\epsilon \rceil} \geq \delta
\cdot w(X).$
\end{proof}

\setcounter{theorem}{\theThmScalingMetricRamseyTheorem}
\begin{theorem}
\ThmScalingMetricRamseyTheorem
\end{theorem}
\setcounter{theorem}{\thetmp}
\begin{proof}
Let $X$ be an arbitrary metric space. Let $Z \subseteq X$. Define
$\ell(Z) = \max \{ \log_{1/\delta} \frac{w(X)}{\bsize(Z)} , 1 \}$.
We will construct a subspace $S(Z) \subseteq Z$ and an ultrametric
$U(S(Z))$ recursively as follows: use the decomposition described
in Section~\ref{sec:ramsey-decomposition} with $\eta = 1/(6
\vartheta(\ell(Z)) )$ and $\DDelta=\ddiam(Z)/2$ to obtain a
partition of the graph $(Q,\bQ)$ and $P \subset Q$ satisfying
Lemma~\ref{lem:decomp}. Run the algorithm on $P$ and $\bQ$
recursively, obtaining subspaces $S(P)$ and $S(\bQ)$ and
ultrametrics $U(S(P))$ and $U(S(\bQ))$ respectively. Let $S(Z) =
S(P)\cup S(\bQ)$. The ultrametric $U(Z)$ is constructed by
creating a root $r$ labeled with $\Lambda(r)=\ddiam(S(Z))$ with
two children at which we root the trees defining $U(S(P))$ and
$U(S(\bQ))$.

We first prove by induction on the size of $Z \subset X$ that
there exists an embedding of $S(Z)$  into $U(S(Z))$ with coarsely
scaling distortion $12\lceil \vartheta(\log_{1/\delta} 1/\epsilon)
\rceil$, and that $\ddiam(U(S(Z)) = \ddiam(S(Z))$. More
specifically let $r_\epsilon(u)$ be the minimum $r$ such that
$w(B(u,r)) \geq \epsilon \cdot w(X)$. That is, we let $G_\epsilon
= \{ (u,v) | d(u,v) \geq \max\{r_{\epsilon/2}(u),r_{\epsilon/2}(v)
\} \}$.

If $Z$ includes a single point $z$ then the claim trivially holds.
Assume by induction that the claim holds for strict subsets of
$Z$. Consider $x,y \in S(Z)$. If $x,y \in S(P)$ their distance in
$U(Z)$ is the same as in $U(P)$ and therefore the claim follows
from the induction hypothesis. If $x,y \in \bQ$ then a similar
argument holds. Let $x \in P$ and $y \in \bQ$, and assume $x,y \in
G_\epsilon$. We may assume that $d(x,y) \leq \ddiam(Z)/4$,
otherwise we can bound the distortion by 4. It follows that
$\bsize(Z) \geq w(B(x, d(x,y)) \geq \epsilon \cdot w(X)$. Then by
Lemma~\ref{lem:decomp} $d(x,y) \geq \eta \DDelta \geq \eta/2 \cdot
\ddiam(S(Z))$. As $d(x,y) \leq \ddiam(S(Z))$ it follows that
$d(x,y) \leq d_{U(Z)}(x,y) \leq 2\eta\cdot d(x,y)$, so that the
distortion of $x$ and $y$ is bounded by $12\lceil
\vartheta(\log_{1/\delta} 1/\epsilon) \rceil$.


We prove by induction on the size of $Z \subseteq X$ that
\[
w^\xx(S(Z)) \geq w(Z)\cdot \delta^{\int_{\ell(Z)}^\infty
\frac{dx}{\vartheta(x)}} .
\]
If $Z=\{z\}$ includes a single point $z$ then the claim trivially
holds since $w^\xx(S(Z))= w(Z)=1$. By applying the induction
hypothesis and Lemma~\ref{lem:decomp} we obtain
\begin{eqnarray*}
w^\xx(S(Z)) & = & w^\xx(S(P)) + w^\xx(S(\bQ)) \geq w(P)\cdot
\delta^{\int_{\ell(P)}^\infty \frac{dx}{\vartheta(x)}} +
w(\bQ)\cdot \delta^{\int_{\ell(\bQ)}^\infty
\frac{dx}{\vartheta(x)}} \\
& \geq & w(Q) \cdot \left(\frac{\bsize(Z)}{\bsize(Q)}\right)^{-1/
\vartheta(\ell(Z)) } \cdot \delta^{\int_{\ell(P)}^\infty
\frac{dx}{\vartheta(x)}} + w(\bQ)\cdot
\delta^{\int_{\ell(\bQ)}^\infty \frac{dx}{\vartheta(x)}}.
\end{eqnarray*}
Since
\begin{eqnarray*}
 \left(\frac{\bsize(Z)}{\bsize(Q)}\right)^{-1/\vartheta(\ell(Z)) } & = &
\delta^{\frac{\log_{1/\delta}\left(\frac{\bsize(Z)}{\bsize(Q)}\right)}{\vartheta(\ell(Z))
}} \geq \delta^{\int_{\ell(Z)}^{\ell(Q)} \frac{dx}{\vartheta(x)}},
\end{eqnarray*}
we get that
\begin{eqnarray*}
w^\xx(S(Z)) & \geq & w(Q) \cdot \delta^{\int_{\ell(Z)}^{\ell(Q)}
\frac{dx}{\vartheta(x)} + \int_{\ell(P)}^\infty
\frac{dx}{\vartheta(x)}} + w(\bQ)\cdot
\delta^{\int_{\ell(\bQ)}^\infty \frac{dx}{\vartheta(x)}} \\
& \geq & (w(Q) + w(\bQ)) \cdot \delta^{\int_{\ell(Z)}^\infty
\frac{dx}{\vartheta(x)}}  = w(Z) \cdot
\delta^{\int_{\ell(Z)}^\infty \frac{dx}{\vartheta(x)}}
\end{eqnarray*}
Noting that $\bsize(X) \leq w(X)$ and $\int_1^\infty
\frac{dx}{\vartheta(x)}=1$ we conclude that $w^\xx(S(X)) \geq
\delta \cdot w(X)$.
\end{proof}

The Ramsey theorems which obtain linear size spaces with local
distortions are discussed in the next section.

\section{Ramsey Embedding, Ramsey Covers and Proximity Data Structures}

In this section we obtain first deterministic constructions of
distance oracles and approximate ranking data structures with
constant query time and asymptotically optimal storage-stretch
tradeoffs.

\subsection{Ramsey Embedding}

The first step is to extend our Ramsey theorems to obtain Ramsey
embeddings (see definition \ref{def:ramsey-embedding}). As our
construction must support distortions which vary as function of
the pairs of points this does not follow directly from the Ramsey
theorems themselves but the Ramsey embeddings can be derived by an
appropriate modification of the algorithms and proofs as described
below. We obtain the following theorems:

\def\ThmMetricRamseyEmbeddingTheorem{
For every $0 < \e$, and any $n$-point metric space $X$, there
exists a Ramsey embedding of $X$ into an ultrametric with a core
subspace of size $n^{1-\e}$ and distortion $O(1/\e)$. }

\begin{theorem}\label{thm:ramsey-embedding}
\ThmMetricRamseyEmbeddingTheorem
\end{theorem}

\def\ThmPartialMetricRamseyEmbeddingTheorem{
For every $\delta >0$ and $\epsilon >0$, and any $n$-point metric
space $X$, there exists a Ramsey $(1-\epsilon)$-partial embedding
of $X$ into an ultrametric with a core subspace $Y$ of size $\geq
\delta n$ and distortion $O(\lceil \log_{1/\delta} 1/\epsilon)
\rceil$. }

\begin{theorem}\label{thm:ramsey-embedding-partial}
\ThmPartialMetricRamseyEmbeddingTheorem
\end{theorem}

\def\ThmScalingMetricRamseyEmbeddingTheorem{ For
every $\delta >0$, any $n$-point metric space $X$, there exists a
Ramsey embedding of $X$ into an ultrametric with a core subspace
$Y$ of size $\geq \delta n$ and scaling distortion $O(\lceil
\vartheta(\log_{1/\delta} 1/\epsilon)) \rceil$. As a consequence
its $\ell_q$-distortion is bounded by $O(\lceil
\vartheta(\log_{1/\delta} q)) \rceil$. }

\begin{theorem}\label{thm:ramsey-embedding-scaling}
\ThmScalingMetricRamseyEmbeddingTheorem
\end{theorem}

To modify the proofs in Section~\ref{sec:ramsey} we need to build
an ultrametric over the entire space $X$ rather than just on the
subspace $S(X)$. Let us describe how to modify the constructions.

It would be useful to use a variation of  Lemma~\ref{lem:decomp} when replacing the function $w$ with $w_C$ which is restricted to a subset $C$. Similar variation can be applied to the theorems of Section~\ref{sec:ramsey} (similarly replacing $w^\phi$ with $w_C^\psi)$.

Let $X$ be an arbitrary metric space. The algorithm builds an
ultrametric recursively. For a $Z \subseteq X$ we also maintain a
core $C(Z) \subseteq Z$. 
Initially $Z=C(Z)=X$. We will construct a
subspace $S(Z) \subseteq C(Z)$ and an ultrametric $U(S(Z))$
recursively as follows: use the decomposition described in
Section~\ref{sec:ramsey-decomposition} on $Z$ with the variation described above for $C=C(Z)$, and with $\eta_Z$
defined appropriately as in the proofs of the theorems in
Section~\ref{sec:ramsey} and $\DDelta=\ddiam(Z)/2$ to obtain a
partition $(Q,Z\setminus Q)$ of the metric space on $Z$, and
a set $P \subset Q \cap C(Z)$ satisfying Lemma~\ref{lem:decomp}. Let $R = \{
x \in Z | d(x,P) \leq \eta_Z \DDelta/2 \}$. Run the algorithm
recursively on $R$ with core $C(R)=P$ and on $\bR$ with core
$C(\bR)=\bQ$, obtaining subspaces $S(R)$ and $S(\bR)$ and
ultrametrics $U(R)$ and $U(\bR)$ respectively. It is easy to
verify that construction above is indeed valid satisfying $P
\subseteq R$ and $\bQ \subseteq \bR$, where the second condition
follows as $d(P,\bQ) > \eta_Z \DDelta$.  
 Let $S(Z) = S(R)\cup S(\bR)$. The ultrametric $U(Z)$ is constructed by creating a root
$r$ labeled with $\Lambda(r)=\ddiam(Z)$ with two children at which
we root the trees defining $U(R)$ and $U(\bR)$.

We prove by induction on the size of $Z \subset X$ that there is a
Ramsey embedding of $Z$ into $U(Z)$ with distortion $4/\eta_Z$.

If $Z$ includes a single point $z$ then the claim trivially holds.
Assume by induction that the claim holds for strict subsets of
$Z$. Consider $x \in S(Z)$ and $y \in Z$. If $x \in S(R)$ and $y
\in R$ their distance in $U(Z)$ is the same as in $U(R)$ and
therefore the claim follows from the induction hypothesis. If $x
\in S(\bR)$ and $y \in \bR$ then a similar argument holds. Let $x
\in S(R) \subseteq P$ and $y \in \bR$ then by definition of $R$,
$d(x,y) \geq \eta_Z \DDelta/2$. Similarly if $x \in S(\bR)
\subseteq \bQ$ and $y \in R$ then by Lemma~\ref{lem:decomp}
$d(x,y) \geq d(P,\bQ) - \eta_Z \DDelta/2 \geq \eta_Z \DDelta/2$.
Hence in both cases $d(x,y) \geq \eta_Z \DDelta/2 \geq \eta_Z
\ddiam(Z)/4 $. As $d(x,y) \leq \ddiam(Z)$ it follows that $d(x,y)
\leq d_{U(Z)}(x,y) \leq 4/\eta_Z\cdot d(x,y)$, and we conclude that
our Ramsey embedding has distortion $4/\eta$. The rest of the
proofs is the same as in Section~\ref{sec:ramsey}.

\subsection{Ramsey Covers}

 Let $\M$ be a class of metric spaces. Assume that given a metric space $X$ we can construct a Ramsey
embedding of $X$ into $Y \in \M$ of size at least $\beta|X|$ which
embeds with distortion $\alpha$ (possibly a function of pairs in
$X$).

We build a Ramsey cover as follows. We apply this construction
iteratively a follows: Let $X_0 = X$ and let $Z_0$ be the core
subspace of $X$, and $Y_0$ be image of $X_0$ in $\M$ under the
Ramsey embedding. For $i>0$ let $X_i = X_{i-1} \setminus Z_{i-1}$
and let $Z_i$ be the core subspace of $X$, and $Y_i$ be the image
of $X_i$ in $\M$ under the Ramsey embedding.

This construction yields a collection of spaces $Y_0, \ldots Y_t
\in \M$, where $|Y_i| \leq (1-\beta)^i |X|$, such that for every
$x \in X$ there exists $i \in [ t ]$ such that for every $y \in
X_i$ the distortion of $x$ and $y$ in $Y_i$ is at most
$\alpha(x,y)$. In particular, for $s\geq 1$, $\sum_i |Y_i|^s \leq
|X|^s/\beta$.

By using Theorems~\ref{thm:ramsey-embedding},
~\ref{thm:ramsey-embedding-partial},
and~\ref{thm:ramsey-embedding-scaling} in the above construction
of Ramsey cover we obtain a cover by ultrametrics.

\subsection{Application to Proximity Data Structures}

An approximate distance oracle is a data structure for a given
metric space $X$ of size $S$ (the space) such that for every $x,y
\in X$ an approximation of the distance between them can be
computed in time $Q$ (the query time) and distortion $D$ (also
called stretch).

Thorup and Zwick \cite{TZ05} gave a randomized construction of
approximate distance oracles of size $O(t\cdot n^{1+1/t})$,
distortion $2t-1$ and query time $O(k)$, for any $k \in
\mathbb{N}$. Mendel and Naor \cite{MN06} gave a different
randomized construction of approximate distance oracles of size
$n^{1+1/t}$, distortion $O(t)$ and query time $O(1)$. The
preprocessing expected time is $O(n^{2+1/t}\log n)$. In
\cite{TZ05} it is shown that this space-distortion tradeoff is
best possible up to the constants.

One can naturally define $(1-\epsilon)$-partial and scaling
distortion distance oracles. These notions have been previously
studied in \cite{ABN06,CDG06} where the results of \cite{TZ05}
have been adapted to accommodate these notions. In particular, the
scaling distortion constructions imply constant average distortion
in $O(n\log n)$ space.

 Using the construction of the Ramsey cover by ultrametrics
described above we can obtain several new {\em deterministic}
constructions of approximate distance oracles. The application
follows since computing the distance in an ultrametric can be done
by computing the LCA of the two leaves in $O(1)$ time. We get the
following results:
\begin{theorem}
There exist deterministic constructions with the following
properties:
\begin{enumerate}
\item For every $t \geq 1$, there exists an approximate distance
oracle with space: $n^{1+1/t}$, distortion: $O(t)$, and query time
$O(1)$.
\item For every $t \geq 1$ and $\epsilon \in (0,1)$, there exists a $(1-\epsilon)$-partial
approximate distance oracle with space: $n\cdot
(\frac{1}{\epsilon})^{1/t}$, distortion: $O(t)$, and query time:
$O(1)$.
\item For every $t \geq 1$ and $\epsilon \in (0,1)$, there exists a
approximate distance oracle with space: $n^{1+1/t}$, distortion:
$O(\lceil \vartheta(t \frac{\log \frac{1}{\epsilon}}{\log n} )
\rceil )$, and query time: $O(1)$.
\end{enumerate}
The preprocessing time can be bounded\footnote{The bound on the
preprocessing time for the algorithm described here is larger by a
factor of $n$. This can be improved by a more involved
implementation and the details are left for the full version.} by
$O(n^2)$.
\end{theorem}

In particular we get the following corollary:
\begin{corollary}
There exists a deterministic construction of an approximate
distance oracle with space: $O(n)$, distortion: $O(\vartheta(\log
\frac{1}{\epsilon}) )$ and query: $O(1)$. In particular it has
average distortion: $O(1)$ and $\ell_q$-distortion\footnote{the
worst case distortion can bounded by $O(\log n)$ by modifying
slightly the definition of $\vartheta$.} $O(\min \{\vartheta(q),
\log n \})$.
\end{corollary}

Similar results can be derived for the approximate ranking problem
via an approach similar to \cite{MN06}.

\subsection{$k$-Local Embeddings and Data Structures}

\setcounter{theorem}{\theThmLocalMetricRamseyTheorem}
\begin{theorem}
\ThmLocalMetricRamseyTheorem
\end{theorem}
\setcounter{theorem}{\thetmp}
\begin{proof}
The proof is similar to the one described in the head of this
section. The main difference is that the core subgraph to be
partitioned $C(Z)$ is chosen as the subspace of $Z$ of maximum
diameter amongst all subspaces of size at most $k$. To take care
that the embedding is non-expansive we set the label of the
constructed ultrametric $U(Z)$ to be $\Delta(C(Z))/(C\cdot t)$ for
some appropriate constant $C$. We omit the details of the proof.
\end{proof}

\setcounter{theorem}{\theThmScalingLocalMetricRamseyTheorem}
\begin{theorem}
\ThmScalingLocalMetricRamseyTheorem
\end{theorem}
\setcounter{theorem}{\thetmp}
\begin{proof}
Again, we follow the proof described in the head of the section.
To obtain scaling local distortion we modify the subspace to be
partitioned $C(Z)$ to be a subspace $H$ that maximizes
$\Delta(H')/|H'|$ over all subspaces $H'$ such that $H' \cap H
\neq \emptyset$. We set the label of the constructed ultrametric
$U(Z)$ to be $\Delta(C(Z))/(C\cdot \{ \min\{t,\vartheta(\log k)
\})$ for some appropriate constant $C$. We omit the details of the
proof.
\end{proof}

From the theorems above we can deduce the following new proximity
data structure results. We give the distance oracles version
below:

\begin{theorem}
There exist deterministic constructions with the following
properties:
\begin{enumerate}
\item For every $t \geq 1$, and $1 \leq k \leq n$ there exists an approximate distance
oracle with space: $n\cdot k^{1/t}$, local distortion: $O(t)$, and
query time $O(1)$.
\item For every $t \geq 1$, and $1 \leq k \leq n$ there exists an approximate distance
oracle with space: $n^{1+1/t}$,  scaling local distortion:
$O(\min\{ \vartheta(\log k), t\})$, and query time $O(1)$.
\end{enumerate}
\end{theorem}

\section{Ramsey Theorems for Decomposable Metric Spaces}

The main idea for achieving better metric Ramsey theorems for
decomposable metric spaces is to reduce the problem to embedding
metric spaces of small aspect ratio. Such metric spaces can be
embedded with low distortion using Rao's method \cite{rao}. Such a
reduction is not quite possible but it is possible to reduce to
the case of hierarchical metric spaces where each level has small
aspect ratio. We will show that such metric spaces have low
distortion embeddings as well. We recall the following definitions
from \cite{BLMN03}:

\begin{definition}[Metric Composition]\label{def:metric-composition}
Let $M$ be a finite metric space. Suppose that there is a
collection of disjoint finite metric spaces $N_x$ associated with
the elements $x$ of $M$. Let $\N = \{ N_x \}_{x \in M}$. For
$\beta\geq 1/2$, the $\beta$-composition of $M$ and $\N$, denoted
by $C=M_\beta[\N]$, is a metric space on the disjoint union $\dot
\cup_x N_x$. Distances in $C$ are defined as follows. Let $x,y \in
M$ and $u \in N_x, v \in N_y$, then:
\[ d_C(u,v)= \biggl \{\begin{matrix} d_{N_x}(u,v) & x=y \\
  \beta \gamma d_M(x,y) & x\neq y .\end{matrix}  \]
where $\gamma=\frac{\max_{z \in M} \diam(N_z)}{\min_{x\neq y \in
M} d_M(x,y)}$.
\end{definition}

\begin{definition}[Composition Closure] \label{def:comp}
Given a class $\M$ of finite metric spaces, we consider
$\comp_\beta(\M)$, its closure under $\ge \beta$-compositions.
Namely, this is the smallest class $\C$ of metric spaces that
contains all spaces in $\M$, and satisfies the following
condition: Let $M \in \M$, and associate with every $x \in M$ a
metric space $N_x$ that is isometric to a space in $\C$. Also, let
$\beta'\geq \beta$. Then $M_{\beta'}[\N]$ is also in $\C$.
\end{definition}

We prove the following general lemma:
\begin{lemma}
\label{lem:composition-embed} Let $\M$ be a class of finite metric
spaces such that every $M \in \M$ has an embedding into $L_p$ with
distortion $\alpha$ then every $X \in \comp_2(\M)$ embeds into
$L_p$ with distortion $2\alpha$.
\end{lemma}
\begin{proof}
Assume every $M \in \M$ has a non-expansive embedding $f_M:M\to
L_p$ with distortion $\alpha$. Let $X \in \comp_2(\M)$. We define
an embedding $\hat{f}_X:X\to L_p$ recursively on the structure of
the metric composition. If $X \in \M$ then $\hat{f}_X=f_X$.
Otherwise let $X = M_{\beta'}[\N]$, $\beta'\geq 2$. Let $u \in X$
such that $u \in N_x$, $x \in M$. Let $\hat{f}_X(u)=(\beta' \gamma
\cdot f_M(x)) \oplus \hat{f}_{N_x}(u)$.

W.l.o.g we may assume that $\| f_M(x) \|_p \leq \diam(M)$. We first
claim by induction that $\|\hat{f}_X(u)\|_p \leq 2\diam(X)$. This
follows as $\|\hat{f}_X(u)\|_p \leq \beta' \gamma \cdot \diam(M) +
2\diam(N_x) \leq \diam(X) + 2\gamma\cdot \diam(M) \leq 2\diam(X)$.

Consider $u,v \in X$. Let $X' = M'_{\beta''}[\N']$, $\beta''\geq 2$, be the first
level in the composition structure such that $u \in N'_x$ and $v
\in N'_y$ for $x \neq y$. Then
\begin{eqnarray*}
\|\hat{f}_X(u)-\hat{f}_X(v)\|_p^p &=&
\|\hat{f}_{X'}(u)-\hat{f}_{X'}(v)\|_p^p =
\|f_{M'}(x)-f_{M'}(y)\|_p^p +
\|\hat{f}_{N_x}(u)-\hat{f}_{N_y}(v)\|_p^p\\
& \leq & \left(\beta''\gamma \cdot d(x,y)\right)^p + \left( 2\max
\{\diam(N_x),\diam(N_y)\} \right)^p \\
& \leq & d(u,v)^p + \left(2\gamma \cdot d(x,y)\right)^p \leq 2
d(u,v)^p .
\end{eqnarray*}

On the other hand
\begin{eqnarray*}
\|\hat{f}_X(u)-\hat{f}_X(v)\|_p & \geq & \|f_{M'}(x)-f_{M'}(y)\|_p
\geq \beta''\gamma \cdot d(x,y)/\alpha = d(u,v)/\alpha .
\end{eqnarray*}
\end{proof}

 Let ${\bf \Phi}$ denote the
class of metric spaces $M$ with aspect ratio at most $\Phi$. Then
we have the following:
\begin{corollary}
\label{cor:composition-Phi}
 Let $X \in \comp_2({\bf \Phi})$ then
$X$ embeds into $L_p$ with distortion $O(\tau_X (\log
\Phi)^{1/p})$.
\end{corollary}
\begin{proof}
Apply Lemma~\ref{lem:composition-embed} on Rao's embeddings for
metric spaces of aspect ratio $\Phi$ to obtain the claimed
distortion bound.
\end{proof}

We are now ready to prove the metric Ramsey theorem:

\setcounter{theorem}{\theThmMetricRamseyTheoremDecomp}
\begin{theorem}
Let $X$ be a metric space. There exists $C>0$ such that for every
$1 \leq p \leq \infty$, and any $\alpha > 1$:
\[
\XX_{L_p}(X,\alpha) \ge 1-C\left(\frac{\tau_X}{\alpha}\right)^p
p\log\left(\frac{\tau_X}{\alpha}\right).
\]
\end{theorem}
\setcounter{theorem}{\thetmp}
\begin{proof}
  Let $\alpha' = \left(\frac{\tau_X}{\alpha}\right)^p/c$ (we may assume $\alpha' \geq
  1$), where $c$ is a constant to be set later.
We first use Theorem~\ref{thm:ramsey-ultrametrics} to obtain a
subspace $X'$ of $X$ which is $\alpha'$ equivalent to an
ultrametric $Y'$ and satisfies the weighted Ramsey condition with
$\XX_\um(X,\alpha') \geq 1-\frac{C}{\alpha'}$.

We be apply Lemma 3.15 of \cite{BLMN03} we obtain a subspace $Y''
\subseteq Y'$ which is $\alpha''$ equivalent to a $k$-HST where
$k=2\alpha'\alpha''$ and $Y''$ satisfies the weighted Ramsey
condition with $\XX_{{\rm k-HST}}(\um,\alpha'') \geq
1-\frac{\log(k/\alpha'')}{\log\alpha''} \geq 1-2\frac{\log
\alpha'}{\log \alpha''} \geq 1-2\frac{\log \alpha'}{\alpha'}$
where we let $\alpha'' = \exp(\alpha')$.

It follows that there exists a subspace $X'' \subseteq X'$ that is
$\alpha'\alpha''$ equivalent to a $k=2\alpha'\alpha''$-HST and
satisfies the weighted Ramsey condition with
$\XX_\um(X,\alpha')\XX_{{\rm k-HST}}(\um,\alpha'') \geq
1-2\frac{\log \alpha'}{\alpha'}$.

Now, applying Lemma 3.16 of \cite{BLMN03} we get that $X''$ is
2-equivalent to a metric space in $\comp_2({\bf \Phi})$, where
$\Phi = \alpha'\alpha'' \leq \exp(2\alpha')$.

Finally, using Corollary~\ref{cor:composition-Phi} we get that
$X''$ is $2c'\tau_X(2\alpha')^{1/p}$ equivalent to a subspace of
$L_p$, where $c'>0$ is a constant. Hence for an appropriate choice
of $c$ we get that $X''$ is $\alpha$ equivalent to a subspace of
$L_p$.
\end{proof}

\section{Deterministic Embeddings}

In \cite{ABN06} a method for embedding finite metric spaces was
developed which unifies many metric embedding results as well as
improving and strengthening some of the known embeddings into
$L_p$. Here, we describe how to construct deterministic embeddings
using our Ramsey decompositions.

Recall the following definition:
\begin{definition}[Partition]
Let $(X,d)$ be a finite metric space. A partition $P$ of $X$ is a
collection of disjoint sets ${\cal C}(P) = \{ C_1, C_2, \ldots,
C_t \}$ such that $X = \cup_j C_j$. The sets $C_j \subseteq X$ are
called clusters. For $x \in X$ we denote by $P(x)$ the cluster
containing $x$. Given $\Delta>0$, a partition is $\Delta$-{\em
bounded} if for all $1 \leq j \leq t$, $\diam(C_j) \leq \Delta$.
\end{definition}

We first construct a bundle of $\hat{\Delta}$-bounded partition of
the metric space by consecutively applying lemma~\ref{lem:decomp}.
However for the purpose of achieving a partition with properties
similar to those of the probabilistic partitions of \cite{ABN06}
we need to choose the ``padding parameter" $\eta$
more carefully. Specifically, we will let $\eta$ depend on the
choice of the point $v$ in the Ramsey decomposition procedure.
That is set $\eta = \log(1/\delta)/\min \{ \log
(\bsize(X)/\bsize(Q)), 2^6 \}$, where $0 < \delta <1$. It follows
that the decomposition creates a partition $(Q,\bQ)$ and a set $P$
such that $w(P) \geq w(Q) \cdot \delta$ and $d(P,Q) \geq \eta\cdot
\hat{\Delta}$. We define $C=Q$ to be a cluster in the partition
and let $\eta(C) = \eta$. We keep applying the lemma on $\bQ$
until $\Delta(\bQ) \leq \hat{\Delta}$. This defines the first
partition in the bundle. We then set $X'$ to be $X$ after all the
core sets $P$ have been removed and repeat the process for $X'$.
This is repeated until $X'$ is empty. It follows that this happens
after repeating $O(\log n/\delta)$ times. We obtain the following
lemma:

\begin{lemma}[Explicit Padded Partitions]\label{lem:det-padded-partition}
  For every $n$-point metric space $X$ it is possible to efficiently
construct deterministically a bundle of  $O(\log n/\delta)$
$\hat{\Delta}$-bounded partitions such that for every $x \in X$
there is a cluster $C$ in the bundle such that $x \in C$ and
$d(x,X/C) \geq \eta(C) \hat{\Delta}$.
\end{lemma}

This replaces the use of the uniform probabilistic partitions in \cite{ABN06}. It can be shown that the properties of these partitions provide similar qualities necessary for their main theorem.
The other randomness used in their proof is for choosing $O(\log
n)$ independent $0,1$ valued random variables for each cluster.
These can be replaces with assigning binary code words of length
$O(\log n)$. Now, using this and
lemma~\ref{lem:det-padded-partition} we can use the framework of
\cite{ABN06} to define the partition-based embeddings of
\cite{ABN06,ABN-doubling, ABN15-local} in order to obtain
deterministic constructions of embeddings into $L_p$. In
particular we get a deterministic embedding of general metric
spaces into $L_p$ in $O(\log n)$ dimensions with $O(\log n)$
distortion, $O(1)$ average distortion, and $O(q)$
$\ell_q$-distortion.

\section{Multi-Embedding}

In this section we obtain multi-embeddings of metric spaces into
ultrametrics with optimal path-distortion. The construction is
analogous to that of Theorem~\ref{thm:ramsey-ultrametrics}
combined with the proof in \cite{bm}. We first give the following
variant of Lemma~\ref{lem:decomp}:
\begin{lemma}
\label{lem:decomp-multi} Given a metric space $X$, and $0 <
\DDelta \leq \ddiam(X)/4$, and integer $t \geq 2$, then there exists a
partition $(Q,\bar{Q})$ of $X$, and $P \subseteq Q$, such that:
$|Q| \leq |X|/2$, $\ddiam(Q) \leq \DDelta$, $d(P,\bQ) \geq \DDelta / (8t)$, and
\begin{eqnarray*}
w(P) \geq w(Q) \cdot \left( \frac{\bsize(X)}{\bsize(Q)}
\right)^{-1/t}.
\end{eqnarray*}
\end{lemma}
\begin{proof}
Let $u,v \in X$ be two points such that $d(u,v) = \ddiam(X)$. Then one of the open balls of radius $\ddiam(X)/2$ around either $u$ or $v$ contains at most $|X|/2$ points. Assume w.l.o.g this holds for $v$ and let $B$ be the associated ball, and let $X'$ be composed of $B$ and the nearest neighbor of $v$ in $X\setminus B$. Then $\Delta(X') \geq \ddiam(X)/2$, so that $\DDelta \leq \Delta(X')/2$. Now, apply Lemma~\ref{lem:decomp} on $X'$. Note that $\Delta(Q) \leq \DDelta < \Delta(X')$, so that $|Q| \leq |X'|-1 \leq |X|/2$.
\end{proof}

\setcounter{theorem}{\theThmMultiEmbedding}
\begin{theorem}
\ThmMultiEmbedding
\end{theorem}
\setcounter{theorem}{\thetmp}
\begin{proof}
Let $X$ be an arbitrary metric space. Let $Z \subseteq X$. We will
construct a multi-embedding of $Z$ into an ultrametric $U(Z)$
recursively as follows: use the decomposition described in
Lemma~\ref{lem:decomp-multi} with $t = \lceil 1/\e \rceil$ and
$\DDelta=\ddiam(Z)/4$ to obtain a partition of the graph $(Q,\bQ)$
and $P \subset Q$ satisfying Lemma~\ref{lem:decomp}. Run the
algorithm on $Q$ and $\bP = Z\setminus P$ recursively, obtaining
multi-embedding into ultrametrics $U(Q)$ and $U(\bP)$
respectively. The ultrametric $U(Z)$ is constructed by creating a
root $r$ labeled with $\Lambda(r)=\ddiam(Z)$ with two children at
which we root the trees defining $U(Q)$ and $U(\bP)$.

Let $\xx = 1+1/t \leq 1+\e$. Next, we prove by induction on the size
of $Z \subseteq X$ that
\[
w^\xx(U(Z)) \leq w(Z)\cdot \bsize(Z)^{\e}.
\]
If $Z=\{z\}$ includes a single point $z$ then the claim trivially
holds since $w^\xx(U(Z))= w(Z) \bsize(Z)^{\e} = w(z)^\xx$. By
applying the induction hypothesis and Lemma~\ref{lem:decomp} we
obtain
\begin{eqnarray*}
w^\xx(U(Z)) & = & w^\xx(U(Q)) + w^\xx(U(\bP)) \leq w(Q)\cdot
\bsize(Q)^{\e} + w(\bP)\cdot \bsize(\bP)^{\e} \\
& \leq & w(P) \cdot \left(\frac{\bsize(Z)}{\bsize(Q)}\cdot
\bsize(Q)\right)^{\e} + w(\bP)\cdot
\bsize(Z)^{\e} \\
& \leq & (w(P) + w(\bP)) \cdot \bsize(Z)^{\e} = w(Z) \cdot
\bsize(Z)^{\e}.
\end{eqnarray*}
Noting that $\bsize(X) \leq w(X)$ we conclude that $w^\xx(U(X))
\leq w(X)^\xx$.

In particular for $w(X)=w^\xx(X) = |X|$ we get that the size of
the ultrametric $U(X)$ is bounded by $n^{1+\e}$.

In addition the multi-embedding we constructed of $X$ into $U(X)$
has the property that the subtrees of $U(X)$ correspond to
subspaces of $X$. At every level of $U(X)$ we have a tree $T=U(Z)$
for some subspace $Z$ and $T$ is split into two subtrees $T_1 =
U(Q)$ and $T_2 = U(\bP)$ defined by the decomposition of
Lemma~\ref{lem:decomp-multi}. Hence $d(Q,\bP) \geq \DDelta/(8t)
= \Delta(Z)/(32t) \geq \e/64 \cdot \Delta(U(Z))$. We also
have that $|P|\leq |Z|/2$ and $\Delta(U(P)) = \Delta(P) \leq
\DDelta = \Delta(Z)/4 = \Delta(U(Z))/4$. In \cite{bm} it is shown
that these properties imply that the path distortion of our
multi-embedding is $O(\min\{\log n,\log \Phi\}/\e)$.
\end{proof}

\bibliographystyle{plain}
\bibliography{art}

\begin{thebibliography}{10}

\bibitem{ABN-doubling}
Ittai Abraham, Yair Bartal, and Ofer Neiman.
\newblock Embedding metric spaces in their intrinsic dimension, 2007.
\newblock Manuscript.

\bibitem{ABN-soda}
Ittai Abraham, Yair Bartal, and Ofer Neiman.
\newblock Embedding metrics into ultrametrics and graphs into spanning trees
  with constant average distortion, 2007.
\newblock To appear in SODA 2007.

\bibitem{ABN07-stoc}
Ittai Abraham, Yair Bartal, and Ofer Neiman.
\newblock Local embeddings of metric spaces.
\newblock In {\em Proceedings of the 39th annual ACM symposium on Theory of
  computing}, STOC '07, pages 631--640, New York, NY, USA, 2007. ACM.

\bibitem{ABN08}
Ittai Abraham, Yair Bartal, and Ofer Neiman.
\newblock Embedding metric spaces in their intrinsic dimension.
\newblock In {\em Proceedings of the 19th annual ACM-SIAM symposium on Discrete
  algorithms}, SODA '08, pages 363--372, Philadelphia, PA, USA, 2008. Society
  for Industrial and Applied Mathematics.

\bibitem{ABN06}
Ittai Abraham, Yair Bartal, and Ofer Neiman.
\newblock Advances in metric embedding theory.
\newblock {\em Advances in Mathematics}, 228(6):3026 -- 3126, 2011.

\bibitem{ABN15}
Ittai Abraham, Yair Bartal, and Ofer Neiman.
\newblock Embedding metrics into ultrametrics and graphs into spanning trees
  with constant average distortion.
\newblock {\em {SIAM} J. Comput.}, 44(1):160--192, 2015.

\bibitem{ABN15-local}
Ittai Abraham, Yair Bartal, and Ofer Neiman.
\newblock Local embeddings of metric spaces.
\newblock {\em Algorithmica}, 72(2):539--606, 2015.

\bibitem{ACEFN18}
Ittai Abraham, Shiri Chechik, Michael Elkin, Arnold Filtser, and Ofer Neiman.
\newblock Ramsey spanning trees and their applications.
\newblock In Artur Czumaj, editor, {\em Proceedings of the Twenty-Ninth Annual
  {ACM-SIAM} Symposium on Discrete Algorithms, {SODA} 2018, New Orleans, LA,
  USA, January 7-10, 2018}, pages 1650--1664. {SIAM}, 2018.

\bibitem{BBM06}
Y.~Bartal, B.~Bollob\'as, and M.~Mendel.
\newblock Ramsey-type theorems for metric spaces with applications to online
  problems.
\newblock {\em Journal of Computer and System Sciences}, 72(5):890--921, August
  2006.
\newblock Special Issue on FOCS 2001.

\bibitem{BLMN03}
Y.~Bartal, N.~Linial, M.~Mendel, and A.~Naor.
\newblock On metric ramsey-type phenomena.
\newblock {\em Annals Math}, 162(2):643--709, 2005.

\bibitem{B96}
Yair Bartal.
\newblock Probabilistic approximation of metric spaces and its algorithmic
  applications.
\newblock In {\em Proceedings of the 37th Annual Symposium on Foundations of
  Computer Science}, FOCS '96, pages 184--193, Washington, DC, USA, 1996. IEEE
  Computer Society.

\bibitem{B98}
Yair Bartal.
\newblock On approximating arbitrary metrices by tree metrics.
\newblock In {\em Proceedings of the 30th annual ACM symposium on Theory of
  computing}, STOC '98, pages 161--168, New York, NY, USA, 1998. ACM.

\bibitem{B04}
Yair Bartal.
\newblock Graph decomposition lemmas and their role in metric embedding
  methods.
\newblock In {\em Algorithms - {ESA} 2004, 12th Annual European Symposium,
  Bergen, Norway, September 14-17, 2004, Proceedings}, pages 89--97, 2004.

\bibitem{B07-icms}
Yair Bartal.
\newblock Advances in metric ramsey theory and their applications.
\newblock Presented at the ICMS Workshop on Geometry and Algorithms, 2007.

\bibitem{B11-course}
Yair Bartal.
\newblock Lecture notes in metric embedding theory and its algorithmic
  applications, 2011.
\newblock URL:
  \url{http://moodle.cs.huji.ac.il/cs10/file.php/67720/GM_Lecture6.pdf}.

\bibitem{BFN19-covers}
Yair Bartal, Nova Fandina, and Ofer Neiman.
\newblock Covering metric spaces by few trees.
\newblock In Christel Baier, Ioannis Chatzigiannakis, Paola Flocchini, and
  Stefano Leonardi, editors, {\em 46th International Colloquium on Automata,
  Languages, and Programming, {ICALP} 2019, July 9-12, 2019, Patras, Greece},
  volume 132 of {\em LIPIcs}, pages 20:1--20:16. Schloss Dagstuhl -
  Leibniz-Zentrum f{\"{u}}r Informatik, 2019.

\bibitem{BFN16}
Yair Bartal, Arnold Filtser, and Ofer Neiman.
\newblock On notions of distortion and an almost minimum spanning tree with
  constant average distortion.
\newblock In {\em Proceedings of the Twenty-Seventh Annual ACM-SIAM Symposium
  on Discrete Algorithms}, SODA '16, pages 873--882, Philadelphia, PA, USA,
  2016. Society for Industrial and Applied Mathematics.

\bibitem{bm}
Yair Bartal and Manor Mendel.
\newblock Multi-embedding and path approximation of metric spaces.
\newblock In {\em SODA '03: Proceedings of the fourteenth annual ACM-SIAM
  symposium on Discrete algorithms}, pages 424--433, Philadelphia, PA, USA,
  2003. Society for Industrial and Applied Mathematics.

\bibitem{BM04}
Yair Bartal and Manor Mendel.
\newblock Dimension reduction for ultrametrics.
\newblock In {\em Proceedings of the 15th annual ACM-SIAM symposium on Discrete
  algorithms}, SODA '04, pages 664--665, Philadelphia, PA, USA, 2004. Society
  for Industrial and Applied Mathematics.

\bibitem{bkrs}
A.~Blum, H.~Karloff, Y.~Rabani, and M.~Saks.
\newblock A decomposition theorem for task systems and bounds for randomized
  server problems.
\newblock {\em SIAM Journal on Computing}, 30(5):1624--1661 (electronic), 2000.

\bibitem{bfm}
J.~Bourgain, T.~Figiel, and V.~Milman.
\newblock On {H}ilbertian subsets of finite metric spaces.
\newblock {\em Israel Journal of Mathematics}, 55(2):147--152, 1986.

\bibitem{CDG06}
T.-H.~Hubert Chan, Michael Dinitz, and Anupam Gupta.
\newblock Spanners with slack.
\newblock In {\em Proceedings of the 14th Conference on Annual European
  Symposium - Volume 14}, ESA'06, pages 196--207, London, UK, UK, 2006.
  Springer-Verlag.

\bibitem{EFN15}
Michael Elkin, Arnold Filtser, and Ofer Neiman.
\newblock Prioritized metric structures and embedding.
\newblock In {\em Proceedings of the Forty-Seventh Annual {ACM} on Symposium on
  Theory of Computing, {STOC} 2015, Portland, OR, USA, June 14-17, 2015}, pages
  489--498, 2015.

\bibitem{FRT03}
Jittat Fakcharoenphol, Satish Rao, and Kunal Talwar.
\newblock A tight bound on approximating arbitrary metrics by tree metrics.
\newblock {\em Journal of Computer and System Sciences}, 69(3):485--497, 2004.

\bibitem{FL21}
Arnold Filtser and Hung Le.
\newblock Clan embeddings into trees, and low treewidth graphs.
\newblock {\em CoRR}, abs/2101.01146, 2021.
\newblock to appear in STOC 2021.

\bibitem{indyk}
Piotr Indyk.
\newblock Algorithmic applications of low-distortion geometric embeddings.
\newblock In {\em Proceedings of the 42nd IEEE symposium on Foundations of
  Computer Science}, FOCS '01, pages 10--33, Washington, DC, USA, 2001. IEEE
  Computer Society.

\bibitem{indyk_dist}
Piotr Indyk and Jiri Matou{\v{s}}ek.
\newblock Low-distortion embeddings of finite metric spaces.

\bibitem{krr}
H.~Karloff, Y.~Rabani, and Y.~Ravid.
\newblock Lower bounds for randomized $k$-server and motion-planning
  algorithms.
\newblock {\em SIAM Journal on Computing}, 23(2):293--312, 1994.

\bibitem{KLMN04}
Robert Krauthgamer, James~R. Lee, Manor Mendel, and Assaf Naor.
\newblock Measured descent: a new embedding method for finite metrics.
\newblock {\em Geometric and Functional Analysis}, 15(4):839--858, 2005.

\bibitem{lemin}
Alex~J. Lemin.
\newblock Isometric embedding of ultrametric (non-{A}rchimedean) spaces in
  {H}ilbert space and {L}ebesgue space.
\newblock In {\em $p$-adic functional analysis (Ioannina, 2000)}, volume 222 of
  {\em Lecture Notes in Pure and Appl. Math.}, pages 203--218. Dekker, New
  York, 2001.

\bibitem{linial}
N.~Linial.
\newblock Finite metric spaces- combinatorics, geometry and algorithms.
\newblock In {\em Proceedings of the ICM}, 2002.

\bibitem{MN06}
Manor Mendel and Assaf Naor.
\newblock Ramsey partitions and proximity data structures.
\newblock In {\em Proceedings of the 47th Annual IEEE Symposium on Foundations
  of Computer Science}, FOCS '06, pages 109--118, Washington, DC, USA, 2006.
  IEEE Computer Society.

\bibitem{NT12}
Assaf Naor and Terence Tao.
\newblock Scale-oblivious metric fragmentation and the nonlinear dvoretzky
  theorem.
\newblock {\em Israel Journal of Mathematics}, 192(1):489--504, 2012.

\bibitem{rao}
Satish Rao.
\newblock Small distortion and volume preserving embeddings for planar and
  euclidean metrics.
\newblock In {\em Proceedings of the fifteenth annual symposium on
  Computational geometry}, SCG '99, pages 300--306, New York, NY, USA, 1999.
  ACM.

\bibitem{TZ05}
Mikkel Thorup and Uri Zwick.
\newblock Approximate distance oracles.
\newblock {\em J. ACM}, 52(1):1--24, 2005.

\end{thebibliography}

\appendix

\end{document}